\documentclass[11pt]{article}

\usepackage{amsmath,amsthm,paralist,bm,yhmath}
\usepackage{amssymb}
\usepackage{epstopdf}
\usepackage{epsfig}
\usepackage{graphicx}
\usepackage{cite}
\usepackage{empheq}
\usepackage[footnotesize]{caption}

\usepackage{hyperref}

\makeatletter
\let\NAT@parse\undefined
\makeatother
\usepackage[sort&compress, numbers]{natbib}

\usepackage{subfigure}
\usepackage{fullpage, upgreek}
\usepackage[usenames]{color}
\newcommand{\Rbb}{\mathbb{R}}

\newcommand{\scp}[2]{\langle #1, #2 \rangle}
\newcommand{\Nbb}{\mathbb{N}}
\newtheorem{theorem}{Theorem}

\newtheorem{lemma}{Lemma}
\newcommand{\inv}[1]{\frac{1}{#1}}
\newcommand{\supp}{{\rm supp}\,}
\newcommand{\tinv}[1]{{\textstyle\frac{1}{#1}}}
\newcommand{\sign}{{\rm sign}\,}

\newcommand{\ie}{{i.e.}, } 
\newcommand{\eg}{{e.g.}, } 

\DeclareMathOperator*{\argmin}{argmin}

\DeclareMathOperator{\Id}{Id}

\newcommand{\norm}[1]{\|#1\|}

\newcommand{\cl}{\mathcal}
\newcommand{\bs}{\boldsymbol}
\newcommand{\bb}{\mathbb}
\newlength{\imgwidth}
\newcommand{\st}{\quad{\rm s.t.}\quad}

\title{On the optimality of a $\ell_1/\ell_1$ solver for sparse signal recovery from sparsely
  corrupted compressive measurements.}

\author{Laurent Jacques\thanks{L.\,J. is funded by the Belgian F.R.S-FNRS. ICTEAM Institute, ELEN Department, Universit\'e catholique de
    Louvain (UCL), B-1348 Louvain-la-Neuve, Belgium. Email:
    laurent.jacques@uclouvain.be}}
\date{\today}

\begin{document}

\maketitle
\begin{abstract}
This short note proves the $\ell_2-\ell_1$ instance optimality of a
$\ell_1/\ell_1$ solver, \ie a variant of \emph{basis pursuit denoising}
with a $\ell_1$ fidelity constraint, when applied to the estimation of
sparse (or compressible) signals observed by sparsely corrupted compressive
measurements.  The approach simply combines two known results due to
Y. Plan, R. Vershynin and E. Cand\`es.  
\end{abstract}

\paragraph*{Conventions:} Most of domain dimensions (\eg $M$, $N$) are
denoted by capital roman letters. Vectors and matrices are associated
to bold symbols while lowercase light letters are associated to scalar
values.  The $i^{\rm th}$ component of a vector $\bs u$ is $u_i$ or
$(\bs u)_i$. The identity matrix is $\Id$. 
The set of indices in $\Rbb^D$ is $[D]=\{1,\,\cdots,D\}$. Scalar product between
two vectors $\bs u,\bs v \in \Rbb^{D}$ reads $\bs u^* \bs v = \scp{\bs
  u}{\bs v}$ (using the transposition $(\cdot)^*$). For any $p\geq 1$,
$\|\cdot\|_p$ represents the $\ell_p$-norm such that $\|\bs u\|_p^p =
\sum_i |u_i|^p$ with $\|\bs u\|=\|\bs u\|_2$ and $\|\bs u\|_\infty =
\max_i |u_i|$. The $\ell_0$ ``norm'' is $\|\bs u\|_0 = \# \supp \bs
u$, where $\#$ is the cardinality operator and $\supp \bs u = \{i: u_i
\neq 0\} \subseteq [D]$.  For $\cl S \subseteq [D]$, $\bs u_{\cl S}\in
\Rbb^{\#\cl S}$ (or $\bs \Phi_{\cl S}$) denotes the vector (resp. the
matrix) obtained by retaining the components (resp. columns) of $\bs
u\in\Rbb^D$ (resp. $\bs \Phi\in\Rbb^{D'\times D}$) belonging to $\cl
S\subseteq [D]$. The operator $\cl H_K$ is the hard thresholding operator setting all the
coefficients of a vector to 0 but those having the $K$ strongest
amplitudes. The set of canonical $K$-sparse signals in $\Rbb^N$ is $\Sigma_K=\{\bs v\in\Rbb^N: \|\bs v\|_0 \leq
K\}$. $B^N_2$ and $S^{N-1}$ are the $\ell_2$ ball
and $(N-1)$-sphere in $\Rbb^N$, respectively. Finally, the operator
$\sign \lambda$, which equals to $1$ if
$\lambda$ is positive and $-1$ otherwise, is applied component wise onto vectors.

\section{Introduction}
\label{sec:introduction}

Let us consider the case where a sparse (or compressible) signal $\bs
x\in\Rbb^N$ is observed with a random Gaussian matrix $\bs \Phi \sim
\cl N^{M \times N}(0,1)$,
\begin{equation}
  \label{eq:2}
  \bs y = \bs \Phi \bs x + \bs n,
\end{equation}
with a sparse (or Laplacian) noise $\bs n$ of bounded $\ell_1$-power,
\ie there exists a bound $\epsilon>0$ such that $\|\bs n\|_1 \leq
\epsilon$ with high (and controlled) probability. 

In this short note, we prove the stability of a variant of the \emph{basis
pursuit denoising} program, namely
\begin{equation}
  \label{eq:1}
  \argmin_{\bs u\in\Rbb^N} \|\bs u\|_1\ \st \|\bs y - \bs\Phi u\|_1
  \leq \epsilon\tag{BPDN-$\ell_1$},
\end{equation}
in estimating $\bs x$ from $\bs y$ under an $\ell_1$-fidelity constraint. The mathematical tools we are
going to use are those developed in the recent work of Y. Plan and
R. Vershynin in the context of 1-bit compressed sensing
\cite{plan2012robust} combined with Cand\`es' simplified proof of
basis pursuit denoising $\ell_2-\ell_1$-instance optimality \cite{candes2008rip}. 
No elements are specially new except their combination. In particular,
it is interesting to see how these two pieces
of works fit nicely in order to reach the announced objective. 

\section{BPDN-$\ell_1$ instance optimality}
\label{sec:bpdn-ell_1-instance}

Here is the main result of this note.

\begin{theorem} 
\label{thm:bpdn-ell_1-instance}
Let $\bs \Phi\in\Rbb^{M\times N}$ be a sensing matrix
  used in \eqref{eq:2} and assume
  that there exist 3 constants $\delta_{2K},\delta_{3K}\in (0,1)$ and $\nu>0$
  such that,
  for all $\bs u\in\Sigma_{2K}$ and $\bs v\in\Sigma_{K}$ with
  $\scp{\bs u}{\bs v}=0$,
\begin{align}
  \label{eq:l1-norm-preserv-2K}
  \big|\tinv{M} \|\bs \Phi \bs u\|_1\  -\
  \nu\|\bs u\| \big|&\leq \delta_{2K} \|\bs u\|,\\
  \label{eq:l1-orthog-preserv-3K}
  \big|\tinv{M} \scp{\sign(\bs \Phi \bs u)}{\bs\Phi\bs v}\big|&\leq \delta_{3K} \|\bs
  v\|.
\end{align}
Then, if $\delta_{2K} + \delta_{3K} \leq \nu -
\tinv{2}$, the solution $\bs x^*$ of BPDN-$\ell_1$ respects
$$
\|\bs x^* - \bs x\|\ \leq\ 8 \frac{\epsilon}{M}\ +\ 12 e_0(K),
$$
with $e_0(K) = \|\bs x - \bs x_K\|_1/\sqrt K$.  
\end{theorem}

Before to prove this theorem, the following lemma (mainly a rewriting
of a result given in \cite{plan2012robust}) assures us on the
feasibility of the conditions \eqref{eq:l1-norm-preserv-2K} and
\eqref{eq:l1-orthog-preserv-3K}.  

\begin{lemma}
Let $N,M,K\in \Nbb$ and $\delta\in [0,1]$. There exist two constants
$C,c>0$ such that, for 
\begin{equation}
\label{eq:cond-on-M}
M \geq C \delta^{-6} K \log (2N/K)
\end{equation}
and $\bs \Phi \sim \cl N^{M\times N}(0,1)$, we have, with a probability at
least $1-8\exp(-c\delta^2 M)$, 
\begin{align}
  \label{eq:l1-norm-preserv}
  \big|\tinv{M} \|\bs \Phi \bs u\|_1\,-\,
  \sqrt{\tfrac{2}{\pi}}\|\bs u\| \big|&\leq \delta \|\bs u\|,\\
  \label{eq:l1-orthog-preserv}
  \big|\tinv{M} \scp{\sign(\bs \Phi \bs u)}{\bs\Phi\bs v}\big|&\leq \delta \|\bs
  v\|,
\end{align}
for all $\bs u, \bs v \in \Sigma_K$ with $\scp{\bs u}{\bs v}=0$.
\end{lemma}
\begin{proof}
Let us write $\cl K=\Sigma_K\cap B_2^N$ and $\cl K^*=\Sigma_K\cap S^{N-1}$. Using \cite[Prop. 4.3]{plan2012robust} with $\tau=0$, we know that
there exist two constants $C,c>0$ such that if 
$$
M \geq C \delta^{-6} K \log (2N/K)
$$
and if $\bs \Phi=(\bs\varphi_1,\cdots,\bs\varphi_M)^T \sim \cl
N^{M\times N}(0,1)$ with $\bs \varphi_i\in\Rbb^N$ ($1\leq i\leq M$),
then, with probability at least $1-8 \exp(-c \delta^2 M)$,
$$
\sup_{\bs a \in \cl K^*,\,\bs b \in \cl K - \cl K} \big|f_{\bs a}(\bs b) -
\bb E f_{\bs a}(\bs b) \big| \leq \delta, 
$$
where $f_{\bs a}(\bs b) := \tinv{M}\sum_j \sign(\scp{\bs\varphi_j}{\bs
  a}) \scp{\bs\varphi_j}{\bs
  b}$.  Knowing that $\bb E f_{\bs a}(\bs b) = \sqrt{\frac{2}{\pi}}\,\scp{\bs
  a}{\bs b}$, this means that, under the same conditions,
$$
\sup_{\bs a \in \cl K^*,\,\bs b \in \cl K - \cl K} \big|\tinv{M}
\scp{\sign(\bs \Phi \bs a)}{\bs\Phi\bs b} -
\sqrt{\tfrac{2}{\pi}}\scp{\bs a}{\bs b}\big| \leq \delta. 
$$

In particular, for any $\bs u,\bs v \in \Sigma_K$, since $\bs u/\|\bs
u\| \in \cl K^*$ and $\bs v/\|\bs
v\| \in \cl K^* \subset \cl K - \cl K$, we have 
$$
|\tinv{M}
\scp{\sign(\bs \Phi \bs u)}{\bs\Phi\bs v} -
\sqrt{\tfrac{2}{\pi}}\|\bs u\|^{-1}\scp{\bs u}{\bs v}| \leq \delta \|\bs v\|. 
$$
Therefore, if $\scp{\bs u}{\bs v} = 0$, $|\tinv{M} \scp{\sign(\bs \Phi \bs u)}{\bs\Phi\bs v}| \leq \delta \|\bs
v\|$, while taking $\bs u=\bs
v$ leads to 
$$
|\tinv{M} \|\bs \Phi \bs u\|_1  -
\sqrt{\tfrac{2}{\pi}}\|\bs u\| | \leq \delta \|\bs u\|. 
$$
\end{proof}

\paragraph*{Remarks on $\delta$:} The dependency in $\delta^{-6}$ in \eqref{eq:cond-on-M} is
probably not optimal and could be improved. This is actually due to
the fact that this lemma is extendable to much more general sets than
$\cl K$ (\eg compressible signals) \cite{plan2012robust}. For having
only \eqref{eq:l1-norm-preserv}, \cite[Lemma 5.3]{plan2011dimension} shows
that a dependency in $\delta^{-4}$ is allowed. Moreover,
\cite{jacques2011dequantizing} shows that \eqref{eq:l1-norm-preserv}
holds of $M\geq M_0$ with $M_0 = O(\delta^{-2} K\log N/K)$. Proving
that \eqref{eq:l1-orthog-preserv} is respected from the same number of
measurements is an open problem.

\begin{proof}[Proof of Theorem \ref{thm:bpdn-ell_1-instance}]
  We follow partially the procedure given in \cite{candes2008rip} with
  an adaption due to the $\ell_1$-norm fidelity of BPDN-$\ell_1$. Let us write $\bs x^*$ the solution of BPDN-$\ell_1$ and $\bs
  x^*=\bs x+ \bs h$. In order to bound the reconstruction error of
  BPDN-$\ell_1$, we have to characterize the behavior of
  $\|\bs x^* - \bs x\| = \|\bs h\|$. 

  We define $T_0 = \supp \bs x_K$ and a partition $\{T_k:1\leq k \leq
  \lceil (N-K)/K\,\rceil\}$ of the support of $\bs h_{T_0^c}$. This
    partition is determined by ordering elements of $\bs h$ off of the
    support of $\bs x_K$ in decreasing absolute value. We have $|T_k|=K$
  for all $k\geq 1$, $T_k\,\cap\,T_{k'}=\emptyset$ for $k\neq k'$, and
  crucially that $|h_{j}|\leq |h_{i}|$ for all $j\in T_{k+1}$ and
  $i\in T_{k}$.
 
  We start from 
  \begin{equation}
    \label{eq-pt2-a}
    \|\bs h\|\ \leq\ \|\bs h_{T_{01}}\|\ +\ \|\bs h_{T_{01}^c}\|,
  \end{equation}
  with $T_{01} = T_0 \cup T_1$, and we are going to bound separately the
  two terms of the RHS. In \cite{candes2008rip}, it is proved that
  \begin{equation}
    \label{eq:candes-compress-bound}
    \|\bs h_{T_{01}^c}\|\leq \sum_{k\geq 2} \|\bs h_{T_k}\| \leq \|\bs h_{T_{01}}\| +
    2e_0(K),
  \end{equation}
  with $e_0(K) = \tfrac{1}{\sqrt{K}} \norm{\bs x_{T_0^c}}_1$.
Therefore, 
\begin{equation}
\label{eq:bound-h-ht01-e0}
  \|\bs h\|\ \leq\ 2\|\bs h_{T_{01}}\|\ + 2 e_0(K).
\end{equation}
Let us bound now $\|\bs h_{T_{01}}\|$. We have 
\begin{multline*}
\|\bs\Phi \bs h_{T_{01}}\|_1\ =\ \scp{\sign(\bs\Phi \bs h_{T_{01}})}{\bs\Phi
  \bs h_{T_{01}}}\
=\ \scp{\sign(\bs \Phi \bs h_{T_{01}})}{\bs \Phi \bs h}\ -\ \sum_{k\geq
  2}\,\scp{\sign(\bs\Phi \bs h_{T_{01}})}{\bs \Phi \bs h_{T_k}}.
\end{multline*}
By H\"older inequality,
$$
\scp{\sign(\bs \Phi \bs h_{T_{01}})}{\bs \Phi \bs h} \leq \|\bs \Phi \bs h\|_1\ \leq\ \|\bs \Phi \bs x - \bs y\|_1 +  \|\bs
\Phi \bs x - \bs y\|_1 \leq 2\epsilon.
$$
For any $k\geq 2$, since $\bs h_{T_{01}}$ and $\bs h_{T_k}$ are $2K$- and
$K$-sparse, respectively, with $\scp{\bs h_{T_{01}}}{\bs h_{T_k}}=0$,
we know from \eqref{eq:l1-orthog-preserv-3K} that 
\begin{equation*}
|\scp{\sign(\bs \Phi \bs h_{T_{01}})}{\bs \Phi \bs h_{T_k}}|\ \leq\
M\,\delta_{3K} \|\bs h_{T_k}\|.
\end{equation*}
Therefore, using \eqref{eq:l1-norm-preserv-2K} and \eqref{eq:candes-compress-bound}, 
\begin{align*}
  M (\nu-\delta_{2K})\|\bs h_{T_{01}}\|&\leq
   \|\bs \Phi \bs h_{T_{01}}\|_1
\leq
  2\epsilon + M \delta_{3K}\,\sum_{k\geq 2}\|\bs h_{T_k}\|\\
&\leq 2\epsilon  + M \delta_{3K} (\|\bs h_{T_{01}}\| + 2 e_0(K)),
\end{align*}
or equivalently,
\begin{align*}
  \|\bs h_{T_{01}}\|&\leq \tfrac{2}{\nu -
    (\delta_{2K} + \delta_{3K})} \big( \frac{\epsilon}{M}\  +\ \delta_{3K}\,e_0(K) \big).
\end{align*}
Using \eqref{eq:bound-h-ht01-e0}, we find,
\begin{align*}
\|\bs h\|\ \leq\  \tfrac{4}{\nu -
    (\delta_{2K} + \delta_{3K})} \frac{\epsilon}{M}\ +\
  4\,\tfrac{\nu + \delta_{3K} - \delta_{2K}}{\nu -
    (\delta_{2K} + \delta_{3K})} \,e_0(K).
\end{align*}
Finally, taking $\delta_{2K} + \delta_{3K} \leq \nu -
\inv{2}$ provides the result.
\end{proof}


\begin{thebibliography}{1}

\bibitem{plan2012robust}
Y.~Plan and R.~Vershynin,
\newblock ``Robust 1-bit compressed sensing and sparse logistic regression: A
  convex programming approach,''
\newblock {\em IEEE Transactions on Information Theory, to appear.}, 2012.

\bibitem{candes2008rip}
E.~Cand{\`e}s,
\newblock ``{The restricted isometry property and its implications for
  compressed sensing},''
\newblock {\em Compte Rendus de l{'}Academie des Sciences, Paris, Serie I},
  vol. 346, pp. 589--592, 2008.

\bibitem{plan2011dimension}
Y.~Plan and R.~Vershynin,
\newblock ``Dimension reduction by random hyperplane tessellations,''
\newblock {\em arXiv preprint arXiv:1111.4452}, 2011.

\bibitem{jacques2011dequantizing}
L.~Jacques, D.~K. Hammond, and J.~M Fadili,
\newblock ``Dequantizing compressed sensing: When oversampling and non-gaussian
  constraints combine,''
\newblock {\em Information Theory, IEEE Transactions on}, vol. 57, no. 1, pp.
  559--571, 2011.

\end{thebibliography}

\end{document}